\documentclass[preprint,12pt]{elsarticle}

\usepackage{color}

\usepackage[ansinew]{inputenc}
\usepackage{graphicx}
\usepackage{amsmath,amsfonts}
\usepackage{booktabs}
  
\allowdisplaybreaks[1]

\newcommand*{\R}{{\mathbb{R}}}

\newtheorem{theorem}{Theorem}

\newtheorem{lemma}[theorem]{Lemma}

\newtheorem{remark}{Remark}
\newenvironment{proof}[1][Proof]{\textit{#1.} }{\ \rule{0.5em}{0.5em}}

\journal{Journal of Computational and Applied Mathematics}

\begin{document}

\begin{frontmatter}



\title{High-order compact finite difference scheme for option pricing in stochastic volatility models}


\author[label1]{Bertram D{\"u}ring\corref{cor1}} 
\cortext[cor1]{Corresponding author}
\address[label1]{Department of Mathematics, University of Sussex, Pevensey II, Brighton, BN1 9QH, United Kingdom.}
\ead{b.during@sussex.ac.uk}

\author[label2]{Michel Fourni\'e}
\address[label2]{Institut de Math\'ematiques de Toulouse \\
Universit\'e de Toulouse et CNRS (UMR 5219), France.}
\ead{michel.fournie@math.univ-toulouse.fr}

\begin{abstract}
We derive a new high-order compact finite difference scheme for option
pricing in stochastic volatility models. 
The scheme is fourth order accurate in space and second order accurate
in time. 
Under some restrictions, theoretical results like unconditional stability 
in the sense of von Neumann are presented. Where the analysis becomes too involved we
validate our findings by a numerical study. Numerical experiments for
the European option pricing problem are presented. We observe fourth order convergence for non-smooth payoff.
\end{abstract}

\begin{keyword}
Option pricing \sep compact finite difference discretizations \sep
mixed derivatives \sep high-order scheme
\MSC 65M06 \sep 65M12 \sep 91B28
\end{keyword}

\end{frontmatter}

% \linenumbers


\section{Introduction}

The traditional approach to price derivative assets or options 
is to specify an asset price process exogenously by a stochastic
diffusion process and then price by no-arbitrage arguments.
The seminal example of this approach is Black \& Scholes'
paper \cite{BlaSch73} on pricing of European-style options.
This approach leads to simple, explicit pricing formulas. However, 
empirical research has revealed that they are not able to explain
important effects in real financial markets,
e.g.\ the volatility smile (or skew) in option prices.

In real financial markets, not only asset returns are subject to risk,
but also the estimate of the riskiness is typically subject to
significant uncertainty. To incorporate such additional source of
randomness into an asset pricing model, one has to introduce a second
risk factor. This also allows to fit higher moments of the asset
return distribution. The most prominent work in this direction is
Heston model \cite{Hes93}.
Such models are based on a two-dimensional stochastic diffusion process with
two Brownian motions with correlation $\rho$, i.e.,
$dW^{(1)}(t)dW^{(2)}(t)=\rho\, dt,$ on a given filtered probability space for
the stock price $S=S(t)$ and the stochastic volatility $\sigma=\sigma(t)$
\begin{align*}
dS(t)& =\bar{\mu} S(t)\,dt +\sqrt{\sigma(t)} S(t)\,dW^{(1)}(t),\\
d\sigma(t)& =a(\sigma(t)) \,dt+b(\sigma(t))\,dW^{(2)}(t),
\end{align*}
where $\bar{\mu}$ is the drift of the stock, $a(\sigma)$ and $b(\sigma)$ are the drift and the diffusion coefficient of
the stochastic volatility.

Application of It\^o's Lemma leads to 
partial differential equations of the following form
\begin{equation}
\label{P0}
 V_t+\frac12 S^2\sigma V_{SS}+\rho b(\sigma) \sqrt{\sigma}S V_{S\sigma}+\frac12  b^2(\sigma)
V_{\sigma\sigma}+a(\sigma)V_\sigma+rSV_S-rV=0, 
\end{equation}
where $r$ is the (constant) riskless interest
rate. Equation \eqref{P0} has to be solved for
$S,\sigma>0,\,0 \leq t \leq T$ and subject to final and 
boundary conditions which depend on the specific option that is to be
priced. 

For some models and under additional restrictions, closed form
solutions to \eqref{P0} can be obtained by Fourier methods (e.g.\
  \cite{Hes93}, \cite{Due09}). 
Another approach is to derive approximate analytic expressions, see e.g.\
\cite{BeGoMi10} and the literature cited therein.
In general, however, ---even in the Heston model \cite{Hes93}  when the
parameters in it are non constant--- equation \eqref{P0} has
to be solved numerically. Moreover, many (so-called American) options
feature an additional early exercise right. Then one has to solve a
free boundary problem which consists of \eqref{P0} and an early
exercise constraint for the option price. Also for this problem one
typically has to resort to numerical approximations.

In the mathematical literature, there are many papers on numerical
methods for option pricing, mostly addressing the one-dimensional case of a single
risk factor and using standard, second order finite difference
methods (see, e.g., \cite{TavRan00} and the references therein). More
recently, high-order finite difference schemes (fourth order in space)
were proposed that use a compact stencil (three points in space). In
the present context see, e.g., \cite{TaGoBh08} for linear and
\cite{DuFoJu04,DuFoJu03,LiaKha09} for fully nonlinear problems. 

There are less works considering numerical methods for option
pricing in stochastic volatility models, i.e., for two spatial
dimensions. Finite difference approaches that are used are
often standard, low order methods (second order in space) and do
provide little numerical analysis or convergence results. Other
approaches include finite element-finite
volume \cite{ZvFoVe98}, multigrid \cite{ClaPar99}, sparse wavelet
\cite{HiMaSc05}, or spectral methods \cite{ZhuKop10}.

Let us review some
of the related finite difference literature. Different efficient
methods for solving the American
option pricing problem for the Heston model are compared in
\cite{IkoToi07}. The article focusses on the treatment of the early
exercise free boundary 
and uses a second order finite difference discretization.
In \cite{HouFou07} different, low order ADI (alternating
direction implicit) schemes are adapted to the Heston model to include
the mixed spatial derivative term. 
While most of \cite{TaGoBh08} focusses on high-order compact scheme
for the standard (one-dimensional) case, in a short remark \cite[Section~5]{TaGoBh08} also the
stochastic volatility (two-dimensional) case is considered. However,
the final scheme there is of second order only due to the low order
approximation of the cross diffusion term.

The originality of the present work consists in proposing a new, {\em
  high-order compact
finite difference scheme\/} for (two-dimensional) option pricing
models with {\em stochastic volatility}. It should be emphasised that
although our presentation is focused on the Heston model, our methodology
naturally adapts to other stochastic volatility models. 
We derive a new compact scheme that is fourth order accurate in space
and second order accurate in time. 
The stability analysis of the scheme is a difficult task due to the 
multi-dimensional context, variable coefficients and the nature of the
boundary conditions.
Under additional assumptions (zero correlation, periodic boundary
conditions), we establish theoretical results like
unconditional stability in the sense of von Neumann (for `frozen
coefficients'). We discuss this in the numerical part.

This paper is organised as follows. In the next section,
we recall the Heston model from \cite{Hes93} and its closed
form solution for the constant parameters case.
In Section~\ref{probsection} we introduce new independent variables
to transform the partial differential equation to a more tractable
form.
In Section~\ref{HOCsection} we derive the new high-order compact
scheme. We analyse its necessary stability condition 
in section \ref{numanalsection}. 
Numerical experiments that confirm
the good properties of the
method are presented in Section~\ref{numsection}. We give
numerical results for the European option pricing
problem with non-smooth payoff and observe fourth order
convergence. Section~\ref{concsection} concludes.

\section{Heston model}

Let us recall the Heston model from \cite{Hes93}
on which we will focus our presentation.
Consider a two-dimensional standard Brownian motion
$W=(W^{(1)},W^{(2)})$ with correlation $dW^{(1)}(t)dW^{(2)}(t)=\rho dt$
on a given filtered probability space.
Assuming a specific form of the drift $a(\sigma)$ and the diffusion
coefficient $b(\sigma)$ of the stochastic volatility, 
the value of the underlying asset in \cite{Hes93} is characterised by 
 \begin{align}
 dS(t)& =\bar{\mu} S(t) \,dt+\sqrt{\sigma(t)} S(t)\,dW^{(1)}(t),\quad \nonumber \\
 \label{SDEs}
 d\sigma(t)& =\kappa^*(\theta^*-\sigma(t)) \,dt+v\sqrt{\sigma(t)}\,dW^{(2)}(t),
 \end{align}
 for $0< t\leq  T$ with $S(0),\sigma(0)>0$ and $\bar{\mu}$, $\kappa^*$, $v$ and $\theta^*$ the
 drift, the mean reversion speed, the volatility of volatility and the
 long-run mean of $\sigma,$ respectively.

 Note that our method
carries over to other stochastic volatility models with
different choices of the drift and the diffusion coefficient of
the stochastic volatility, e.g., the GARCH diffusion model
\begin{equation}
\label{eq:garchmodel}
d\sigma(t)=\kappa^*(\theta^*-\sigma(t))\,dt+v\sigma(t)\,dW^{(2)}(t),
\end{equation}
or the so-called 3/2-model
\begin{equation}
\label{eq:32model}
d\sigma(t)= \kappa^*\sigma(t) (\theta^*-\sigma(t))\,dt+v{\sigma(t)}^{3/2}\,dW^{(2)}(t),
\end{equation}
in a natural way (see also Remark \ref{otherstochmodels} at the end of
section~\ref{HOCsection:derivation}).

 In the Heston model, it follows by It\^o's lemma and standard arbitrage arguments that any
 derivative asset $V=V(S,\sigma,t)$ solves the following partial differential equation
 \begin{multline}
  V_t+\frac12 S^2\sigma V_{SS}+\rho v {\sigma}S V_{S\sigma}+\frac12 v^2\sigma
  V_{\sigma\sigma}+rSV_S\\
 +\big[\kappa^* (\theta^*-\sigma)-\lambda(S,\sigma,t)\big]V_\sigma-rV=0, \label{P1}
 \end{multline}
 which has to be solved for $S,\sigma>0$, $0 \leq t < T$ and subject to a suitable
 final condition, e.g.,
 $$V(S,\sigma,T)=\max(K-S,0),$$
 in case of a European put option (with $K$ denoting the strike price).
 In \eqref{P1}, $\lambda(S,\sigma,t)$ denotes the market price of volatility risk.
 While in principle it could be estimated from market data, this is
 difficult in practice and the results
 are controversial. Therefore, one typically assumes a risk premium
 that is proportional to $\sigma$ and chooses $\lambda(S,\sigma,t)=\lambda_0\sigma$ for some
 constant $\lambda_0$. For streamlining the presentation we restrict ourselves
 to this important case, although our scheme applies to
 general functional forms $\lambda=\lambda(S,\sigma,t).$

 The `boundary' conditions in the case of the put option read as
 follows
 \begin{subequations}
 \begin{align}
   V(0,\sigma,t)&=Ke^{-r(T-t)},&  &T> t\geq 0,\;\sigma>0,
 \label{boundary1}\\
   V(S,\sigma,t)&\to 0,& &T> t\geq 0,\;\sigma>0,\; \text{as } S\to\infty,
   \label{boundary2}\\
   V_\sigma(S,\sigma,t)&\to 0,& &T> t\geq 0,\;S>0,\; \text{as } \sigma \to\infty.
   \label{boundary4}
 \end{align}
 The remaining boundary condition at $\sigma=0$ can be obtained by looking at
 the formal limit $\sigma\to 0$ in \eqref{P1}, i.e.,
 \begin{equation}
   V_t+rSV_S+\kappa^*\theta^* V_\sigma-rV= 0,\quad T> t\geq 0,\;S>0,\; \text{as } \sigma\to  0.
   \label{boundary3}
 \end{equation}
 This boundary condition is used frequently, e.g.\ in \cite{IkoToi07,ZvFoVe98}.
 Alternatively, one can use a homogeneous Neumann condition
 \cite{ClaPar99}, i.e.,
 \begin{equation}
   V_\sigma(S,\sigma,t) \to 0, \quad T> t\geq0,\;S>0,\; \text{as } \sigma\to  0.
 \end{equation}
\end{subequations}

 For {\em constant\/} parameters, one can employ Fourier transform techniques and
 obtain a system of ordinary differential equations which can be
 solved analytically \cite{Hes93}.
 By inverting the transform one arrives at a closed-form solution of
 \eqref{P1}, where the European put option price $V$ is given by
 \begin{equation}
 \label{HestonFormula}
 V(S,\sigma,t)=Ke^{-r(T-t)}\mathcal{I}_2 -S\mathcal{I}_1, 
 \end{equation}
 with ($k=1,2$)
 \begin{align}
 \mathcal{I}_k&=\frac12+\frac1\pi\int_0^\infty\mathrm{Re}\biggl[\frac{e^{-i\xi
     \ln( K)}f_k(\xi)}{i\xi}\biggl]\,d\xi,\label{I1I2}\\
  f_k(\xi)&=\exp\big(C(T-t,\xi)+\sigma D(T-t,\xi)+i\xi \ln S\big),\nonumber\\
 C(\tau,\xi)&=r\xi i \tau+\frac{\kappa^*\theta^*}{v^2}\Bigl[(b+d)\tau -2\ln\Bigl(\frac{1-ge^{d\tau}}{1-g}\Bigr)\Bigr],\;
 D(\tau,\xi)= \frac{b_k+d_k}{v^2}  \frac{1-{e^{d_k\tau}}}{1-g{e^{d_k\tau}}},\nonumber \\
 g&=\frac{b_k+d_k}{b_k-d_k},\quad
 d_k=\sqrt { \left( {\xi}^{2}\mp i\xi \right) { v}^{2}+ b_k^{2},
 }\quad
 b_k=\kappa^*+\lambda_0-\rho v(i\xi+\delta_{1k}).\nonumber
 \end{align}
Here, $\delta_{i,j}$ denotes Kronecker's delta.

 \section{Transformation of the equation and boundary conditions}
 \label{probsection}

 Under the transformation of variables 
 \begin{equation}
   \label{trafo}
   x=\ln \Big(\frac SK\Big),\quad \tilde t=T-t,
   \quad u=\exp(r\tilde t)\frac VK,
 \end{equation}
 (we
 immediately drop the tilde in the following) we arrive at
 \begin{multline}
 \label{P2}
 u_t-\frac12 \sigma \bigl(u_{xx}+2\rho vu_{x\sigma}+v^2u_{\sigma\sigma}\bigr)\\
 +\Big(\frac12 \sigma-r \Big)u_x-\big[\kappa^*\theta^*-(\kappa^* +\lambda_0)\sigma\big]u_\sigma=0,
 \end{multline}
 which is now posed on $\R\times\R^+\times(0,T).$
 We study the problem using the modified parameters
 $$
 \kappa=\kappa^*+\lambda_0,\quad \theta=\frac{\kappa^*\theta^*}{\kappa^*+\lambda_0},
 $$
 which is both convenient and standard practice. For similar reasons,
 some authors set the market price of volatility risk to zero.
 Equation \eqref{P2} can then be written as
 \begin{equation}
 \label{P3}
 u_t-\frac12 \sigma \bigl(u_{xx}+2\rho vu_{x\sigma}+v^2u_{\sigma\sigma}\bigr)+\Big(\frac12 \sigma-r\Big)u_x-\kappa \big[\theta -\sigma\big]u_\sigma=0.
 \end{equation}
 The problem is completed by the following initial and boundary conditions:
 \begin{align}
 u(x,\sigma,0) &=\max (1-\exp (x),0),& & x\in\R ,\;\sigma>0,\nonumber \\
 u(x,\sigma,t) &\to 1,& & x\to -\infty ,\;\sigma>0,\;t>0,\nonumber \\
 u(x,\sigma,t) &\to 0,& & x\to +\infty ,\;\sigma>0,\;t>0,\nonumber \\
 u_\sigma(x,\sigma,t) &\to 0,& & x\in\R ,\;\sigma\to \infty,\;t>0,\nonumber\\
u_\sigma(x,\sigma,t) &\to 0,& & x\in\R ,\;\sigma\to 0,\;t>0.\nonumber
 \end{align}

 %%%%%%%%%%%%%%%%%%%%%%%%%%%%%%%%%%%%%%%%%%%%%%%%%%%%%%%%%%%%%%%%%%% 

 \section{High-order compact scheme}
 \label{HOCsection}

 For the discretization, we replace $\R$ by $[ -R_1,R_1] $
 and $\R^+$ by $[L_2,R_2]$ with $R_1,R_2>L_2>0$ .
 For simplicity, we consider a uniform grid $Z=\{ x_{i}\in \left[ -R_1,R_1%
 \right]:$ $x_{i}=ih_1$, $i=-N,\dots,N\}\times\{ \sigma_{j}\in \left[L_2,R_2%
 \right]:$ $\sigma_{j}=L_2+jh_2$, $j=0,\dots,M\}$ consisting of $(2N+1)\times
 (M+1)$ 
 grid points,
 with $R_1=Nh_1,$ $R_2=L_2+Mh_2$ and with space steps $h_{1}$, $h_{2}$ and time step $k$. Let $u_{i,j}^{n}$
 denote the approximate solution of \eqref{P3} in $(x_{i},\sigma_j)$ at the time $%
 t_{n}=nk$ and let $u^{n}=(u_{i,j}^{n})$. 
 
We impose artificial boundary conditions in a classical manner rigorously studied for a class of  
 Black-Scholes equations in \cite{KanNic00}.
 The boundary conditions on the grid are treated as follows.
Due to the compactness of the scheme, the treatment of the Dirichlet
boundary conditions is minimal.
It is straightforward to consider Dirichlet boundary conditions
without introduction of numerical error by imposing 
 \begin{equation*}
 u_{-N,j}^{n}=1-e^{rt_{n}-Nh}, \quad u_{+N,j}^{n}= 0,\quad (j=0,\dots,M).
 \end{equation*}
At the other boundaries we impose homogeneous Neumann boundary
 conditions.
The treatment of homogeneous Neumann conditions requires more
attention. Indeed, no values are prescribed.
The values of the unknown on the boundaries must be set by extrapolation from values in the interior. Then 
a numerical error is introduced, and the main consideration is that
the order of extrapolation should be high enough not to affect the overall order of accuracy. We refer to the paper of Gustafsson \cite{GusBC} to discuss the influence of the 
order of the approximation on the global convergence rate and justify
our choice of fourth order extrapolation formulae. 
By Taylor expansion, if we cancel the first derivates on the boundaries, it is trivial to verify
%%$$u_{i,-1}=5u_{i,0}-10u_{i,1}+10u_{i,2}-5u_{i,3}+u_{i,4}$$
\begin{equation*}
u_{i,0}^{n}=\frac{18}{11}u_{i,1}^{n}-\frac{9}{11}u_{i,2}^{n}+\frac{2}{11}u_{i,3}^{n},\quad
(i=-N+1,\dots,N-1),
\end{equation*}
%% 4th order, hom. Neumann
%%u0=4/3u1-1/3u2 %% 3rd order, hom. Neumann
%%u0=u1 %% 2nd order, hom. Neumann (Gustaffson: parabolic pb, can
%%recover 4th order inside)
and
\begin{equation*}
u_{i,M}^n = \frac{18}{11} u_{i,M-1}^n - \frac{9}{11} u_{i,M-2}^n +
\frac{2}{11} u_{i,M-3}^n, \quad (i=-N+1,\dots,N-1).
\end{equation*}

 \subsection{Derivation of the high-order scheme for the elliptic problem}
\label{HOCsection:derivation}

  First we introduce the high-order compact finite difference
  discretization for the stationary, elliptic problem with Laplacian operator 
  which appears after the variable transformation $y=\sigma/v$.
  Equation \eqref{P3} is then reduced to the two-dimensional elliptic equation
  \begin{equation}
  \label{eq:convection}
  -\frac{1}{2} v y  (u_{xx}+u_{yy}) - \rho v y u_{xy}+\Big(\frac12 v
  y-r\Big)u_x-\kappa \frac{\theta -vy}{v}u_y=f(x,y),
  \end{equation}
  with the same boundary conditions.\\
  %
  % condition of ellipticity to explicit .......................
  %

  The fourth order compact finite difference scheme
  uses a nine-point computational stencil using the eight nearest neighbouring 
  points of the reference grid point $(i,j).$
 
  The idea behind the derivation of the high-order compact scheme is to operate
  on the differential equations as an auxiliary relation to obtain finite difference
  approximations for high-order derivatives in the truncation error. Inclusion
  of these expressions in a central difference method for  
  equation~(\ref{eq:convection}) increases the order of accuracy, typically to
  $\mathcal{O}(h^4),$ while retaining a compact stencil defined by nodes surrounding a 
  grid point.

  Introducing a uniform grid with mesh spacing $h=h_1=h_2$ in both the $x$- and $y$-direction,
  the standard central difference approximation to equation~(\ref{eq:convection})
  at grid point $(i,j)$ is
  \begin{multline}
  \label{eq:central}
  -\frac{1}{2} v y_j  \bigl(\delta_x^2u_{i,j}+\delta_y^2u_{i,j}\bigr) - \rho v y_j
  \delta_x\delta_y u_{i,j}\\
  +\Big(\frac12 vy_j-r\Big)\delta_x u_{i,j}-\kappa \frac{\theta -vy_j}{v}\delta_y u_{i,j}
    - \tau_{i,j}= f_{i,j},
  \end{multline} 
  where $\delta_x$ and $\delta_x^2$ ($\delta_y$ and $\delta_y^2$, respectively) denote 
  the first and second order central difference approximations with
  respect to $x$ (with respect to $y$). The associated
  truncation error is given by
  \begin{multline}
  \label{eq:tau}
  \tau_{i,j} =
  \frac{1}{24}vyh^{2}
  (u_{xxxx} + u_{yyyy})  +\frac{1}{6}\rho vy h^{2}(u_{xyyy} + u_{xxxy})\\ 
  +\frac{1}{12}( 2\,r-vy ) h^{2} u_{xxx} +\frac{1}{6}{\frac {\kappa  ( \theta -vy ) }{v}}h^{2}u_{yyy} +\mathcal{O}(h^4).
  \end{multline}
  For the sake of readability, here and in the following we omit the
  subindices $j$ and $(i,j)$ on $y_j$ and $u_{i,j}$ (and its
  derivatives), respectively.
  We now seek second-order approximations to the derivatives appearing in (\ref{eq:tau}). 
  Differentiating equation~(\ref{eq:convection}) once with respect to
  $x$ and $y,$ respectively, yields
  \begin{align}
  \label{eq:dx1}
  u_{xxx}=& -u_{xyy} -2\rho u_{xxy} -{\frac {  2r+vy  }{vy}}u_{xx}+2\,{\frac
    {\kappa  ( vy-\theta  )}{{v}^{2}y}}u_{xy}-{\frac {2}{vy}}f_x,\\
  \nonumber
  u_{yyy} =& -u_{xxy}-2\rho u_{xyy} -\frac{1}{y} u_{xx}-\frac{2 \kappa  ( \theta -vy) +v^2}{v^2y} u_{yy} \\
  \label{eq:dx2}
  & \hspace*{2.5cm}  -\frac{2 \rho  +2r- vy}{vy} u_{xy}+\frac{1}{y} u_x+\frac{2 \kappa }{vy}u_y-\frac{2}{vy}f_y.
  \end{align}
  Differentiating equations (\ref{eq:dx1}) and (\ref{eq:dx2}) with respect to $y$ and $x,$ respectively, and adding the two expressions, we obtain
  \begin{multline}
  \label{eq:dxy}
  u_{{{\it xyyy}}}+u_{{{\it xxxy}}}=\frac{vy+2r}{2v{y}^{2}}u_{xx}+\frac{\kappa  (\theta +vy)}{v^2y^2} u_{xy}-
  \frac{4\kappa (\theta -vy)+v^2}{2v^2y}u_{xyy}\\
  -\frac{\rho v+2r-vy}{vy}u_{xxy}-2\rho u_{xxyy}-\frac {1}{2y}u_{xxx}+\frac{1}{vy^2}f_{x} - \frac{2}{vy}f_{xy}.
  \end{multline}
  Notice that all the terms in the right hand sides of (\ref{eq:dx1})-(\ref{eq:dxy})
  have compact $\mathcal{O}(h^2)$ approximations at node $(i,j)$ using finite differences based on $\delta_x$, $\delta_x^2$,  
  $\delta_y$, $\delta_y^2$. We have, for example, ${u_{xxy}}_{i,j}=\delta_x^2\delta_y u_{i,j}+\mathcal{O}(h^2).$
  By differentiating equation~(\ref{eq:convection}) twice
  with respect to $x$ and $y$, respectively,
  and adding the two expressions, we obtain
  \begin{multline}
  \label{eq:dxxyy}
  u_{{{\it xxxx}}}+u_{{{\it yyyy}}}=-2 \rho u_{{{\it xyyy}}}-2 \rho
  u_{{{\it xxxy}}}-2 u_{{{\it xxyy}}}+2{\frac { 
 (   \kappa  vy- {v}^{2}- \kappa  \theta  ) }{{
  v}^{2}y}}u_{{{\it xxy}}} \\
  -{\frac { ( 2 r-vy  ) }{vy}}u_{{{\it xxx}}}+2{\frac { ( \kappa  vy- {v}^{2}-\kappa  \theta    ) }{{v}^{2}y}}u_{{{\it yyy}}}
  -{\frac { ( -vy+4  \rho v+2 r ) }{vy}}u_{{{\it xyy}}} \\
  +4 {\frac {\kappa   }{vy}}u_{{{\it yy}}} + \frac {2}{y}u_{{{\it xy}}}
 - {\frac {2}{vy}}(f_{{{\it xx}}}+f_ {{{\it yy}}}).
  \end{multline}
  Again, using (\ref{eq:dx1})-(\ref{eq:dxy}), the right hand side can be approximated up to $\mathcal{O}(h^2)$ within
  the nine-point compact stencil. 
  Substituting equations~(\ref{eq:dx1})-(\ref{eq:dxxyy}) 
  into equation~(\ref{eq:tau}) and simplifying yields a new
  expression for the 
  error term $\tau_{i,j}$ that consists only of terms which are either
  \begin{itemize}
  \item terms of order $\mathcal{O}(h^4)$, or
  \item terms of order $\mathcal{O}(h^2)$ multiplied by derivatives of $u$ which 
  can be approximated up to $\mathcal{O}(h^2)$ within the nine-point
  compact stencil.
  \end{itemize}
  Hence, substituting the central $\mathcal{O}(h^2)$ approximations to the
  derivatives in this new expression for the error term
  % (\ref{eq:tau2})  
  and inserting it into (\ref{eq:central}) yields the
  following $\mathcal{O}(h^4)$ approximation to the initial partial
  differential equation (\ref{eq:convection}),
   \begin{align}
   &-\frac{1}{24}\frac {h^2((vy_j-2r)^2-4\rho vr-2\kappa ( vy_j-\theta ) -2v^2)+12v^2y_j^2  }{vy_j}{\delta^2_x u_{i,j} }\nonumber \\
   &-\frac{1}{12}\frac {h^2(2\kappa ^{2}(vy_j-\theta )^2-\kappa  v^3y_j-\kappa  \theta v^2
       -v^4)+6v^4y_j^2 }{v^3y_j}{\delta^2_y u_{i,j}}\nonumber\\
   &-\frac{1}{12}h^{2}vy_j(1+2\rho^2) {\delta^2_x\delta^2_y u_{i,j}}\nonumber\\
   &+\frac{h^2}{6}\frac { (\kappa (vy_j-\theta )+v\rho(vy_j-2r) )}{v} {\delta^2_x\delta_y u_{i,j}}\nonumber\\
   &+\frac{{h}^{2}}{12}\frac { (4\kappa \rho(vy_j-\theta )+v(vy_j-2r))}{v} {\delta_x
   \delta^2_y u}_{i,j}\nonumber\\
   &-\frac16 \frac{h^2(\kappa (vy_j-2r)(vy_j-\theta )-\kappa v^2y_j\rho-v^3\rho-v^2r) +6v^3y_j^2\rho}{{v}^{2}y_j}\delta_x \delta_y u_{i,j}\nonumber\\
   &+\frac{1}{12}\frac{6v^2y_j^2-12vy_jr-h^2[v^2+\kappa (vy_j-\theta )]}{vy_j}{\delta_x u_{i,j}}\nonumber\\
   &+\frac{\kappa }{6}{\frac {(6v^2y_j^2-6vy_j\theta - h^2[v^2+\kappa (vy_j-\theta )]) }{{v}^{2}y_j}}\delta_y u_{i,j}\nonumber\\
   =& f_{i,j}
   +\frac{{h}^{2}}{6}{\frac {\rho}{v}} \delta_x\delta_y f_{i,j}
   -\frac{{h}^{2}}{6}{\frac {( {v}^{2}+\kappa ( vy_j-\theta ))}{{v}^{2}y_j}}  \delta_y f_{i,j}\nonumber \\
   &-\frac{{h}^{2}}{12}{\frac { (2\rho v-2r+vy_j)}{vy_j}}  \delta_x
       f_{i,j} 
   \label{eq:scheme2}
   + \frac{{h}^{2}}{12} \delta_x^2 f_{i,j} +\frac{{h}^{2}}{12} \delta_y^2 f_{i,j}. 
   \end{align}
  The fourth order compact finite 
  difference scheme (\ref{eq:scheme2}) considered at the mesh point
  $(i,j)$ involves the nearest eight
  neighbouring mesh points. Associated to the shape of the
  computational stencil, we introduce indexes for each node from zero to nine,
\begin{equation}
 \label{eq:coeffnumber}
  \left (
  \begin{array}{ccc}
  \begin{array}{rcl}
  u_{i-1,j+1}=u_6\\
  u_{i-1,j}=u_3\\ 
   u_{i-1,j-1}=u_7\\ 
  \end{array}
  &
  \begin{array}{rcl}
   u_{i, j+1}=u_2 \\
   u_{i, j}=u_0 \\
   u_{i, j-1}=u_{4} \\ 
  \end{array}
  &
  \begin{array}{rcl} 
  u_{i+1,j+1}=u_5\\
  u_{i+1,j}=u_1\\
u_{i+1,j-1}=u_8\\  
  \end{array}
  \end{array}
  \right ).
  \end{equation}
With this indexing, the scheme (\ref{eq:scheme2}) is defined by
  \begin{equation}
  \label{eq:stencil2}
  \sum_{l=0}^8 \alpha_l u_l = \sum_{l=0}^8 \gamma_l f_l,
  \end{equation}
  where the coefficients $\alpha_l$ and $\gamma_l$ are given by
 \begin{align*}
  \alpha_0=&\bigg( {\frac {4 {\kappa }^{2}+{v}^{2}}{12v}}-{\frac {v
  (2 {\rho}^{2}-5 ) }{3{h}^{2}}} \bigg) y_j\\
 &-{\frac {
 \kappa  {v}^{2}+2 {\kappa }^{2}\theta +{v}^{2}r}{3{v}^{2}}}+{\frac 
 {-{v}^{4}+{\kappa }^{2}{\theta }^{2}-{v}^{3}r\rho+{v}^{2}{r}^{2}}{3{v}^{3
 }y_j}},\\
 \alpha_{1,3}=&\bigg( -\frac v{24}+{\frac {\pm \frac 16v\mp\frac 13\kappa  \rho}{h}}+{\frac {v
  ( \rho^2-1 )  }{3{h}^{2}}} \bigg) y_j
 \mp\frac {\kappa  h}{24}+\frac {\kappa } {12}+\frac r6\\
 &\mp{\frac {v r-\kappa  \theta  \rho}{3vh}}\mp{\frac { ( {v}^{2}-\kappa  \theta 
  ) h}{24vy_j}}-{\frac {-2 rv\rho+\kappa  \theta +2 {r}^{2}-{v}
 ^{2}}{12vy_j}} ,\\
 \alpha_{2,4}=&\bigg( -{\frac {{\kappa }^{2}}{6v}}+{\frac {\pm\frac 13\kappa \mp\frac 16\rho
  v}{h}}+{\frac {v ( \rho^2-1 )  }
 {3{h}^{2}}} \bigg) y_j\mp{\frac {{\kappa }^{2}h}{12v}}+{\frac {
 \kappa   ( {v}^{2}+4 \kappa  \theta  ) }{12{v}^{2}}}\\
 &\mp{
 \frac {rv\rho-\kappa  \theta }{3vh}}\mp{\frac {\kappa  
  ( {v}^{2}-\kappa  \theta  ) h}{12{v}^{2}y_j}}+{\frac {
  ( 2 \kappa  \theta +{v}^{2} )  ( {v}^{2}-\kappa  \theta  ) }{12{v}^{3}y_j}},\\
 \alpha_{5,7}=&\bigg( -\frac {\kappa }{24}\pm{\frac { ( 2 \rho+1 )  ( 
 2 \kappa +v ) }{24h}}-{\frac {v ( \rho+1 ) 
  ( 2 \rho+1 ) }{12{h}^{2}}} \bigg) y_j\\&+{\frac {\kappa  
  ( \rho v+2 r+\theta  ) }{24v}}
 \mp{\frac { ( 2 
 \rho+1 )  ( \kappa  \theta +vr ) }{12vh}}+{\frac {
 {v}^{2}r+{v}^{3}\rho-2 r\kappa  \theta }{24{v}^{2}y_j}},\\
 \alpha_{6,8}=& \bigg( \frac {\kappa }{24}\pm{\frac { ( 2 \rho-1 )  ( 
 -2 \kappa +v ) }{24h}}-{\frac {v ( 2 \rho-1 ) 
  ( \rho-1 ) }{12{h}^{2}}} \bigg) y_j\\
&-{\frac {\kappa  
  ( \rho v+2 r+\theta  ) }{24v}}
 \mp{\frac { ( 2 
 \rho-1 )  ( vr-\kappa  \theta  ) }{12vh}}-{\frac {{v}^{2}r+{v}^{3}\rho-2
  r\kappa  \theta }{24{v}^{2}y_j}},
 \end{align*}
 and
   \begin{align*}
   \gamma_0 & = \frac{2}{3},\qquad
   \gamma_5 =\gamma_7=\frac{\rho}{24},&  \gamma_6 &=\gamma_8=-\frac{\rho}{24},\\
   \gamma_{1,3} &= \frac1{12}\mp \frac{h}{24}\pm\frac1{12}{\frac { ( r-\rho v
       ) h}{vy_j}},&
   \gamma_{2,4} &= \frac1{12}\mp\frac1{12}{\frac{\kappa 
       h}{v}}\mp\frac1{12}{\frac { ({v}^{2}  -\kappa \theta  )
       h}{{v}^{2}y_j}}.  
   \end{align*}
   When multiple indexes are used with $\pm$ and $\mp$ signs, the first index corresponds to the upper sign.

\begin{remark}
\label{otherstochmodels}
The derivation of the scheme in this section can be modified to accomodate other stochastic
volatility models as, e.g., the GARCH diffusion model
\eqref{eq:garchmodel} or the 3/2-model \eqref{eq:32model}.
Using these models the structure of the partial differential equations \eqref{P1}, \eqref{P3}
and \eqref{eq:convection} remains the same, only the coefficients of
the derivatives have to be modified 
accordingly. Similarly, the coefficients of the derivatives in
\eqref{eq:dx1}-\eqref{eq:dxxyy} have to be modified. Substituting these in the modified expression
for the truncation error one obtains equivalent $\mathcal{O}(h^4)$
approximations as \eqref{eq:scheme2}.
\end{remark}

\subsection{High-order scheme for the parabolic problem}

  The high-order compact approach presented in the previous section
  can be extended to the parabolic problem directly by
  considering the time derivative in place of $f(x,y)$.
  Any time integrator can be implemented to solve the problem as presented in
  \cite{SpotzCarey}.
  We consider the most common class of methods involving two times steps.
  For example,
  differencing at  $t_{{\mu}}=(1-{\mu})t^n + {\mu} t^{n+1}$, where
  $0 \leq {\mu} \leq 1$ and the superscript $n$ denotes the time
  level, yields a class of integrators that include the forward Euler
  ($\mu = 0$), Crank-Nicolson ($\mu=1/2$) and backward
  Euler ($\mu = 1$) schemes.
  We use the notation $\delta^+_t  u^n = \frac{u^{n+1}-u^{n}}{k}$. Then
  the resulting fully discrete difference scheme for node $(i,j)$ at the time
  level $n$ becomes
  $$
   \sum_{l=0}^8 \mu \alpha_l u_l^{n+1} + (1-\mu) \alpha_l u_l^{n}   =
  \sum_{l=0}^8 \gamma_l \delta^+_t u_l^n,
  $$
  that can be written in the form (after multiplying by $24 v^3h^2yk$)
  \begin{equation}
    \label{eq:hocscheme}
    \sum_{l=0}^8 \beta_l u_l^{n+1}  = \sum_{l=0}^8 \zeta_l u_l^n.
  \end{equation}
The coefficients $\beta_l,$ $\zeta_l$ are numbered according to the
indexes \eqref{eq:coeffnumber} and are given by
\begin{align*} 
 \beta_0 =& (   (   ( 2 {y_j}^{2}-8  ) {v}^{4}+  (   ( -
8 \kappa -8 r  ) y_j-8 \rho r  ) {v}^{3}+  ( 8 {\kappa 
}^{2}{y_j}^{2}+8 {r}^{2}  ) {v}^{2}\\
&-16 {\kappa }^{2}\theta  vy_j +8 {\kappa }^{2}{\theta }^{2}  ) \mu k+16 {v}^{3}y_j  ) {h}^{2}+
  ( -16 {\rho}^{2}+40  ) {y_j}^{2}{v}^{4}\mu k\\
\beta_{1,3} =&\pm (   (\kappa  \theta  {v}^{2} -{v}^{4}-\kappa  y_j{v}^{3}
  ) \mu k-  ( y_j+2 \rho  ) {v}^{3}+2 {v}^{2}r
  ) {h}^{3}+  (   (   ( -{y_j}^{2}+2  ) {v}^{4}\\
&+  (   ( 4 r+2 \kappa   ) y_j+4 \rho r  ) {v}^{3}-
  ( 2 \kappa  \theta +4 {r}^{2}  ) {v}^{2}  ) \mu k+2
 {v}^{3}y_j  ) {h}^{2}\\
&\pm ( 4 {v}^{4}{y_j}^{2}+  ( -8 {y_j}^{
2}\kappa  \rho-8 y_jr  ) {v}^{3}+8 y_j\kappa  \theta  \rho {v}^{2}
  ) \mu kh+  (8 {\rho}^{2}-8  )
 {y_j}^{2}{v}^{4}\mu k,\\
\beta_{2,4} =& \pm (   ( 2 {\kappa }^
{2}\theta  v-2 {\kappa }^{2}{v}^{2}y_j-2 {v}^{3}\kappa   ) \mu k-2 {v}^{2}y_j\kappa +2 v\kappa  \theta -2 {v
}^{3}  ) {h}^{3}+  (   ( 2 {v}^{4}\\
&+2 \kappa  y_j{v}^{3}+  ( -4 {\kappa }^{2}{y_j}^{2}+2 \kappa  \theta   ) {v}^{2}+8 {
\kappa }^{2}\theta  vy_j-4 {\kappa }^{2}{\theta }^{2}  ) \mu k+2 {v
}^{3}y_j  ) {h}^{2}\\
&\pm  (   ( 8 {y_j}^{
2}\kappa +8 y_j\rho r  ) {v}^{3}-4 {v}^{4}{y_j}^{2}\rho-8 {v}^{2}y_j\kappa  \theta 
  ) \mu kh+  ( 8 {\rho}^{2}-8  ) {y_j}^{2}{v}^{4}\mu k,\\
 \beta_{5,7} =& (   ( {v}^{4}\rho+  ( -{y}^{2}\kappa +\kappa  y_j\rho+r
  ) {v}^{3}+  ( \theta +2 r  ) \kappa  y_j{v}^{2}-2 r
\kappa  \theta  v  ) \mu k\\
&+{v}^{3}\rho y_j  ) {h}^{2}\pm  (   ( 2 \rho+1  ) {y_j}^{2}{v}^{4}+  (   ( 2+4 
\rho  ) \kappa  {y_j}^{2}+  ( -4 \rho r-2 r  ) y_j
  ) {v}^{3}\\
&+  ( -2 \theta -4 \theta  \rho  ) \kappa  y_j{
v}^{2}  ) \mu kh+  ( -2-4 {\rho}^{2}-6 \rho  ) {y_j}^{2
}{v}^{4}\mu k,\\
\beta_{6,8} =&  (   ( -{v}^{4}\rho+  ( {y_j}^{2}\kappa -\kappa  y_j\rho-r
  ) {v}^{3}+  ( -\theta -2 r  ) \kappa  y_j{v}^{2}+2 r
\kappa  \theta  v  ) \mu k\\
&-{v}^{3}\rho y_j  ) {h}^{2}
\pm  (   ( 2 \rho-1  ) {y_j}^{2}{v}^{4}+  (   ( 2-4 
\rho  ) \kappa  {y_j}^{2}+( 2 r-4 \rho r  ) y_j
  ) {v}^{3}\\
&+  ( 4 \theta  \rho-2 \theta   ) \kappa  y_j{v
}^{2}  ) \mu kh+  ( -4 {\rho}^{2}+6 \rho-2  ) {y_j}^{2}
{v}^{4}\mu k,
 \end{align*}
 and
\begin{align*} 
 \zeta_0 =& 16v^3y_jh^2+(1-\mu)k( (   ( 8-2 {y_j}^{2}  ) {v}^{4}+  (   ( 8 \kappa 
 +8 r  ) y_j+8 \rho r  ) {v}^{3}\\
 &+ ( -8 {r}^{2}-8 {
 \kappa }^{2}{y_j}^{2}  ) {v}^{2}+16 {\kappa }^{2}\theta  vy_j-8 {
 \kappa }^{2}{\theta }^{2}  ) {h}^{2}+  ( -40+16 {\rho}^{2}
   ) {y_j}^{2}{v}^{4}
 ),\\
 \zeta_{1,3} =&\pm(2r-(y_j+2\rho)v)v^2h^3+2v^3y_jh^2+(1-\mu)k ( \pm (
 {v}\kappa  y_j+  {v}^{2} -\kappa  \theta      ){v}^{2}{h}^{3}\\
 &+  ( {v}^{2}{y_j}^{2}-  ( 4 r+2 \kappa 
   ) vy_j+   4 {r}^{2}+2 \kappa  \theta -2 {v}^{2}-4 
 \rho vr   ) {v}^{2}  {h}^{2}\\
 &\pm  (   ( -4 {v}+8
  \kappa  \rho  ) {v}^{3}{y_j}^{2}+  ( -8 \kappa  \theta  
 \rho+8 vr  ) {v}^{2}y_j  ) h+  ( 8 {v}^{2}-8 {v}^{2}{
 \rho}^{2}  ) {v}^{2}{y_j}^{2}),
  \\
 \zeta_{2,4} =&\pm(2v\kappa \theta -2v^2y_j\kappa -2v^3)h^3+2v^3y_jh^2+(1-\mu)k ( \pm 2(
 {v}^{3}\kappa - {\kappa }^{2}\theta  v\\
 &+ {\kappa }^{2}{v}^{
 2}y_j  ) {h}^{3}+  ( 4 {\kappa }^{2}{v}^{2}{y_j}^{2}- ( 2 {v}^{2}+8 {\kappa }\theta   )\kappa  v y_j
 +2 \kappa  \theta (2 {\kappa }{\theta }- {v}^{2})-2 {v}^{4}  ) {h}^{2}\\
 &\pm  (   ( -8 {v}^{3}\kappa +4 {v}^{4}\rho  ) {y_j}^{2}+
   ( 8 \kappa  \theta  {v}^{2}-8 {v}^{3}\rho r  ) y_j
   ) h+  ( -8 {v}^{4}{\rho}^{2}+8 {v}^{4}  ) {y_j}^{2}),\\
 \zeta_{5,7} =&v^3\rho y_jh^2+(1-\mu)k (  ( {v}^{3}{y_j}^{2}\kappa -v  ( v\kappa  \theta +2 r\kappa  v+
 \kappa  {v}^{2}\rho  ) y_j\\
 &-v  ( {v}^{2}r-2 r\kappa  \theta +{v}
 ^{3}\rho  )   ) {h}^{2}\pm  ( -v  ( 2 {v}^{3}\rho+{v}
 ^{3}+4 \kappa  {v}^{2}\rho+2 {v}^{2}\kappa   ) {y_j}^{2}\\
 &+v  ( 
 2 v\kappa  \theta +4 v\kappa  \theta  \rho+4 {v}^{2}\rho r+2 {v}^
 {2}r  ) y_j  ) h+v  ( 2 {v}^{3}+6 {v}^{3}\rho+4 {v}^{3
 }{\rho}^{2}  ) {y_j}^{2}), \\
 \zeta_{6,8} =&-v^3\rho y_jh^2 +(1-\mu)k (  ( -{v}^{3}{y_j}^{2}\kappa +v  ( v\kappa  \theta +2 r\kappa  v+
 \kappa  {v}^{2}\rho  ) y_j\\
 &+v  ( {v}^{2}r-2 r\kappa  \theta +{v}
 ^{3}\rho  )   ) {h}^{2}
 \pm  ( v  ( -2 {v}^{3}\rho+{v}
 ^{3}+4 \kappa  {v}^{2}\rho-2 {v}^{2}\kappa   ) {y_j}^{2}\\
 &+v  ( 
 2 v\kappa  \theta -4 v\kappa  \theta  \rho+4 {v}^{2}\rho r-2 {v}^{
 2}r  ) y_j  ) h+v  ( 2 {v}^{3}-6 {v}^{3}\rho+4 {v}^{3}
 {\rho}^{2}  ) {y_j}^{2}).
\end{align*}
When multiple indexes are used with $\pm$ and $\mp$ signs, the first index
corresponds to the upper sign.
Choosing $\mu=1/2,$ i.e., in the Crank-Nicolson case, the
resulting scheme is of order two in time and of order four in space.

\subsection{Stability analysis}
\label{numanalsection}

Besides the multi-dimensionality the initial-boundary-value problem
\eqref{eq:hocscheme} features two main difficulties for its stability analysis:
the coefficients are non-constant and the boundary conditions are not periodic.
In this section, we consider the von Neumann stability analysis (see,
e.g., \cite{StrikwerdaBook})
even if the problem considered does not satisfy periodic
boundary conditions. This approach is extensively used in the literature and
yields good criteria on the robustness of the scheme. 
Other approaches which take into account the boundary conditions 
like normal mode analysis \cite{GKS} are beyond the scope of the 
present paper (we refer to \cite{FournieRigal} for normal mode
analysis for a high-order compact scheme).

To consider the variable coefficients, the principle of `frozen coefficients'
(the variable coefficient problem is stable if all the `frozen' problems are stable) \cite{GKS,StrikwerdaBook}
is employed. It should be noted, that in the discrete case, this principle is far from trivial. 
The most general statements are given in \cite{GKS, Magnus, Wade, StrikewedaWade} and reference therein for hyperbolic problems.
For parabolic problems in the discrete case we refer to \cite{RicMor67, Widlund65}.
Using the frozen coefficients approach gives a necessary stability
condition and slightly strengthened stability for frozen coefficients
is sufficient to ensure overall stability \cite{RicMor67}. 
%Here again, even if 
%unfortunately the theory does not take boundaries into account we bypass this limitation.
%Although supposing periodic boundary-conditions and frozen coefficients, von Neumann analysis involves too high
%complexity to succeed in theoretical analysis, so results are given
%for $\rho=0$ only and perform additional 
%numerical tests to validate the stability for $\rho \neq 0$.

We now turn to the von Neumann stability analysis.
 We rewrite $u^n_{i,j}$ as
\begin{equation}
  \label{eq:Unwave}
  u^n_{i,j}=g^n e^{Iiz_1 + Ijz_2},
\end{equation}
where $I$ is the imaginary unit,
$g^n$ is the amplitude
at time level $n$, and $z_1={2\pi h}/{\lambda_1}$ and $z_2={2\pi h}/{\lambda_2}$ are
phase angles with wavelengths $\lambda_1$ and  $\lambda_2,$ in the
range $[0,2\pi[$, respectively. Then the scheme is stable if for all $z_1$ and $z_2$ the amplification
factor $G={g^{n+1}}/{g^{n}}$ satisfies the relation 
\begin{equation}
  \label{eq:vNstab}
  |G|^2 - 1 \leq 0.
\end{equation}
An expression for $G$ can be found using (\ref{eq:Unwave}) in
(\ref{eq:hocscheme}).

Our aim is to prove von Neumann stability (for `frozen coefficients')
without restrictions on the time step size. 
To show that \eqref{eq:vNstab} holds we would need to study the
(formidable) expression for the amplification factor $G$ (not given
here) which consists of polynomials of order up to six in 13 variables.
To reduce the high number of parameters in the following numerical
analysis, we assume here zero
interest rate $r=0$ and choose the parameter
${\mu}=1/2$ (Crank-Nicolson case). 
Even then, at present a complete analysis for non-zero correlation seems out of
reach, but we are able to show the following result.

\begin{theorem}
\label{thm:stability}
For $r=\rho=0$ and $\mu=1/2$ (Crank-Nicolson), the 
scheme (\ref{eq:hocscheme})
satisfies the stability condition \eqref{eq:vNstab}.
\end{theorem}

\begin{proof}
Let us define new variables
\begin{align*}
c_1=\cos\left(\frac{z_1}{2}\right),&\quad c_2=\cos\left(\frac{z_2}{2}\right),\quad
s_1=\sin\left(\frac{z_1}{2}\right),\quad s_2=\sin\left(\frac{z_2}{2}\right),\\ W&=\frac {2 \left( \theta-vy \right) }{v}s_2,\quad 
V= \frac {2vy}{\kappa}s_1,
\end{align*}
which allow us to express $G$ in terms of $h,k,\kappa,V,W$ and
trigonometric functions only. This reduces the number of variables in
the amplification factor from ten to nine.
The new variable $V$ has constant positive sign contrary to $W$.

In the new variables the stability criterion
\eqref{eq:vNstab} of the 
scheme can be written as 
\begin{equation}
\label{eq:vNstab2}
 \frac{-8kh^2(n_4h^2+n_2)}{d_6h^6 + d_4h^4 + d_2h^2 + d_0} \leq 0, 
\end{equation}
with
\begin{align*}
n_4 & =-4\,V{\kappa}^{3}{\it f_3}\,s_1^{3}{W}^{2}-{V}^{3}{\kappa}^{3}{
\it f_4}\,s_1^{3},\quad
n_2 = -4\,{V}^{3}{\kappa}^{3}{\it f_2}\,{\it f_1}\,{\it s_1},\\
d_6 & = 4\, \left( -2\,W{\it c_2}+V{\it c_1} \right) ^{2}{\kappa}^{2}s_1^{4},\\
d_4 &= \frac{1}{4}\,{\kappa}^{4}s_1^{4} \left( {V}^{2}-4\,V{\it c_1}\,W{\it 
c_2}+4\,{W}^{2} \right) ^{2}{k}^{2}\\
&\quad -4\,V{\kappa}^{3}s_1^{3} \left( {\it f_4}\,{V}^{2}+4\,{\it f_3}
\,{W}^{2} \right) k +16\,{\kappa}^{2}{V}^{2}f_2^{2}s_1^{2},\\
d_2 &= {V}^{2}{\kappa}^{4}s_1^{2} \left( {V}^{2}{\it f_6}-36\,V{\it 
c_1}\,W{\it c_2}+4\,{\it f_5}\,{W}^{2} \right) {k}^{2}-16\,{V}^{3}{
\kappa}^{3}{\it f_2}\,{\it f_1}\,{\it s_1}\,k,\\
d_0 & =4\,{V}^{4}{\kappa}^{4}f_1^{2}{k}^{2},
\end{align*}

where $f_1,$ $f_2,$ $f_3,$ $f_4,$ $f_5,$ and $f_6$  have constant sign and are defined by
\begin{align*}
f_1 &= 2c_1^{2}c_2^{2}+c_1^{2}+c_2^{2}-4 \leq 0,&
f_2 &=c_1^{2}+c_2^{2}+1\geq 0, \\
f_3 &=2c_1^{2}c_2^{2}-c_1^{2}-1 \leq 0,&
f_4 &=2c_1^{2}c_2^{2}-c_2^{2}-1 \leq 0,\\
f_5 &=4c_1^{4}c_2^{2}-2c_1^{2}-c_2^{2}+8 \geq 0,&
f_6 &= 4c_1^{2}{{c_2}}^{4}-2c_2^{2}-c_1^{2}+8\geq 0.
\end{align*}
%We observe that the cosinus and sinus functions who have influence appear with even and odd exponent, respectively.
%Considering that the trigonometric functions assume values in $[-1,1]$. 
We observe that we can restrict our analysis (expect for $d_2$, treated  below)  to the trigonometric functions $s_1$, $s_2$, $c_1,$ and $c_2$ in the reduced range $[0,1]$ (${z_1}/{2}$ and ${z_2}/{2}$ are in $[0,\pi[$,  even exponents for cosinus functions). 
It is straight-forward to verify that
$n_4,$ $n_2,$ $d_6,$ $d_4,$ and $d_0$ are positive. It remains to prove
$d_2=d_{22}k^2 + d_{21}k$ is positive as well. Indeed,
$d_{21}\geq 0$ and $d_{22}$ is a polynomial of degree two in $W$ having a
positive leading order coefficient. The minimum value of $d_{22}$ is
given by
$$
m=2{V}^{4}{\kappa}^{4}s_1^{2}f_1f_7/f_5
$$
with
$f_7 =4c_2^{4}c_1^{4}-2c_1^{4}c_2^{2}-2c_1^{2}c_2^{4}+6c_1^{2}c_2^{2}+c_1^{2}+c_2^{2}-8\leq0.$
Hence, $m$ is positive and then $d_2$ is positive as well.
Therefore, the numerator in \eqref{eq:vNstab2} is negative and the
denominator in \eqref{eq:vNstab2} is positive
which completes the proof.
\end{proof}\\

For non-zero correlation the situation becomes more
involved. Additional terms appear in the expression for the
amplification factor $G$ and we face an additional degree of freedom
through $\rho$. 
Since we have proven condition \eqref{eq:vNstab} for $\rho=0$ it seems
reasonable to assume it also holds at least for
values of 
$\rho$ close to zero. In practical applications, however, correlation can be
strongly negative. Few theoretical results can be obtained, we recall 
the following lemma from \cite{DuFo12}.
\begin{lemma}
\label{partStabLemma}
For  any $\rho$, $r=0$, and $\mu=1/2$ (Crank-Nicolson) it holds: if either 
$c_1=\pm1$ or $c_2=\pm 1$ or $y=0$, then the stability condition \eqref{eq:vNstab} 
is satisfied.
\end{lemma}
\begin{proof}
  See Lemma~1 in \cite{DuFo12}.
\end{proof}

In \cite{DuFo12}, we have reformulated condition  \eqref{eq:vNstab} into a constrained
optimisation problem and have employed a line-search global-optimisation
algorithm to find the maxima.
We have found that
the stability condition \eqref{eq:vNstab} was always satisfied. The maxima for each $\rho\in [-1,0]$ 
were always negative but very close to zero.
This result is in agreement with Lemma~\ref{partStabLemma} (in fact, $|G|^2-1=0$ for $y=0$).
Our conjecture from these results is that the stability condition
\eqref{eq:vNstab} is satisfied also for non-vanishing correlation
although it will be hard to give an analytical proof.

In our numerical experiments we
observe stability also for a general choice of parameters. To validate
the stability property of the scheme also 
for general parameters, we
perform additional numerical tests in section~\ref{numsection}.
 
%%%%%%%%%%%%%%%%%%%%%%%%%%%%%%%%%%%%%%%%%%%%%%%%%%%%%%%%%%%%%%%%%%%%%%%

 \section{Numerical results}
 \label{numsection}

\subsection{Numerical convergence}

 \begin{table}[!t]
 \begin{center}
 \begin{tabular}[c]{l c}
 \toprule
 Parameter & Value\\
 \midrule
  strike price & $K=100$\\
  time to maturity & $T=0.5$\\
  interest rate & $r=0.05$\\
  volatility of volatility & $v=0.1$\\ 
  mean reversion speed & $\kappa =2$\\
 % % NOTE: = "real kappa" + lambda = mean reversion speed + risk premium
 long-run mean of $\sigma$ & $\theta =0.1$\\
 % NOTE: = "real kappa" * "real theta" / ("real kappa" + lambda)
 correlation & $\rho=-0.5$\\
 \bottomrule
 \end{tabular}
 \caption{Default parameters for numerical simulations.}
 \label{defaulttable}
 \end{center}
 \end{table}

In this section we perform a numerical study to compute the order of
convergence of the scheme \eqref{eq:hocscheme}.
Due to the compact discretization the
resulting linear systems have a good sparsity pattern and can be solved
very efficiently.
We compute the $l_2$ norm error $\varepsilon_2$ and the maximum norm error $\varepsilon_\infty$
of the numerical solution with respect to a numerical reference solution on a
fine grid. We fix the parabolic mesh ratio $k/h^2$ to
a constant value which is natural for parabolic PDEs and our
scheme which is of order $\mathcal{O}(k^2)$ in time and
$\mathcal{O}(h^4)$ in space. 
Then, asymptotically, we expect these
errors to converge as $\varepsilon = Ch^m$
for some $m$ and $C$ representing constants. This implies
$\ln(\varepsilon) = \ln(C) + m \ln(h) .$
Hence, the double-logarithmic plot $\varepsilon$ against $h$ should be
asymptotic to a straight line with slope $m$. This gives a method for
experimentally determining the order of the scheme. 

Figure \ref{fig:sol} shows the numerical
solution for the European option price at time $T=0.5$ using the
parameters from Table~\ref{defaulttable}.
\begin{figure}[!ht]
  \centering
  \includegraphics[width=0.65\textwidth]{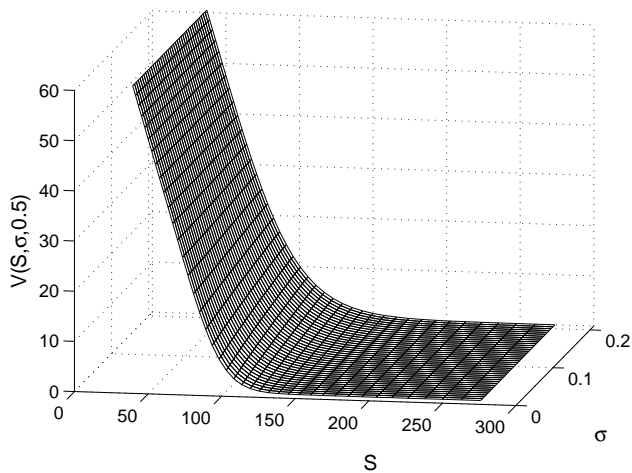}%
  \caption{Numerical solution for the European option
    price.}% 
  \label{fig:sol}
\end{figure}

We refer to
Figure~\ref{fig:numconv1} and Figure~\ref{fig:numconv2} for the
results of the numerical convergence study using the default parameters from
Table~\ref{defaulttable}. For the parameter $\mu$, we use a Rannacher
time-stepping choice \cite{Ran84}, i.e., we start with four fully
implicit quarter time steps 
($\mu=1$) and then continue with Crank-Nicolson ($\mu=1/2$).
\begin{figure}[htb]
  \centering
  \includegraphics[width=0.65\textwidth]{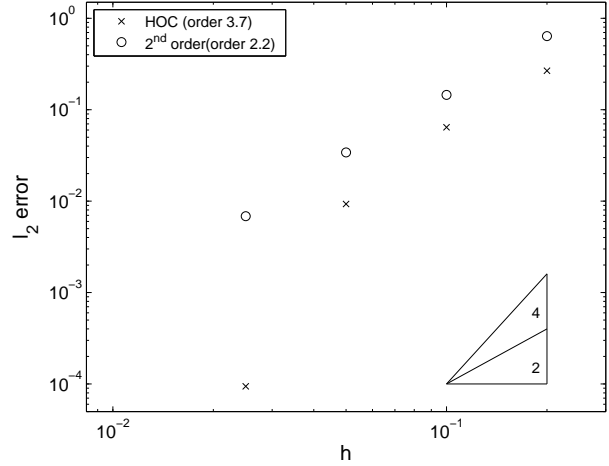}%
  \caption{$l_2$-error vs. $h$.}%
  \label{fig:numconv1}
\end{figure}
\begin{figure}[htb]
  \centering
  \includegraphics[width=0.65\textwidth]{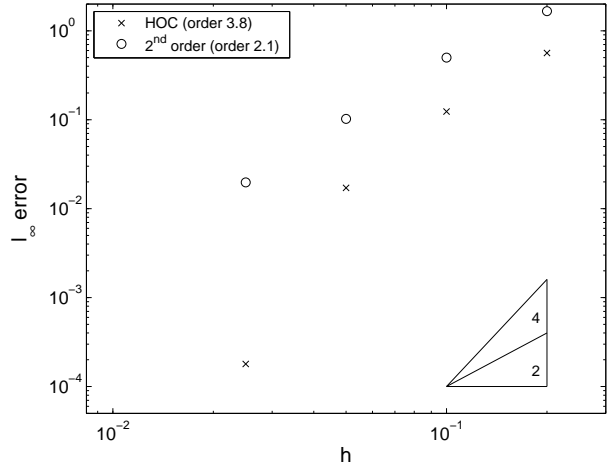}%
  \caption{$l_\infty$-error vs. $h$.}%
  \label{fig:numconv2}
\end{figure}
For comparison we conducted additional
experiments using a standard, second order scheme (based on the central
difference discretization \eqref{eq:central} where we neglect the
truncation error).
We observe that the numerical convergence order agrees well with the
theoretical order of the schemes.
It is important to choose the mesh in such a way that the singular
point of the initial condition is not a point of the mesh.
The construction of such a mesh is always possible in a simple manner. 
Then the non-smooth payoff can be directly considered in our
scheme and we observe fourth order numerical convergence. 

\begin{remark}
Without constraint on the mesh, i.e.\ when then singular point of
the payoff is a mesh point, the rate of convergence is reduced to
two. However, it is possible to recover the fourth order convergence with
such a mesh if the initial data are smoothed.
\end{remark}

The numerical convergence analysis also shows the superior efficiency of the
high-order scheme compared to a standard second order
discretization. 
In each time step of each scheme a linear system has to be solved.
For both schemes this requires the same computational time for the same dimension.
To achieve the same level of accuracy the new scheme
requires significantly less grid points, or in other words, the computational time
to obtain a given accuracy level is greatly reduced by using the
high-order scheme.  

\subsection{Numerical stability analysis}

In our numerical analysis in section~\ref{numanalsection}, we have proven
the stability result Theorem~\ref{thm:stability} for $r=\rho=0.$
To validate this property for general
parameters, we perform additional numerical tests. We compute
numerical solutions for varying values
of the parabolic mesh ratio $k/h^2$ and the mesh width $h.$ Plotting
the associated $l_2$ norm errors in the plane should allow us to detect
stability restrictions depending on $k/h^2$ or oscillations that occur
for high cell Reynolds number (large $h$). This approach for a numerical
stability study was also used in \cite{DuFoJu03}.
\begin{figure}[htb]
  \centering
  \includegraphics[width=0.65\textwidth]{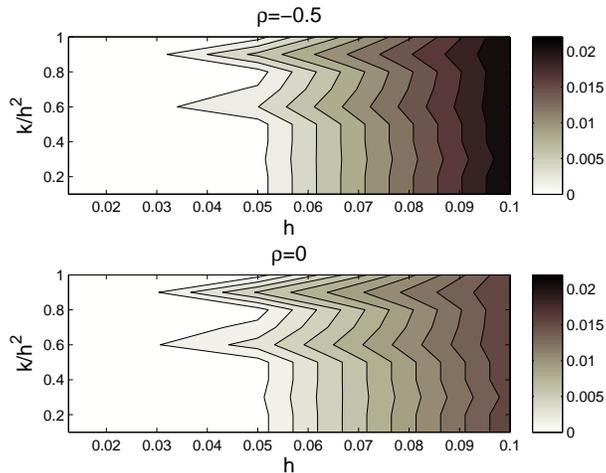}%
  \caption{$l_2$ norm error in the $k/h^2$-$h$-plane for $\rho=-0.5$ (top) and
    $\rho=0$ (bottom).}%
  \label{fig:numstability}
\end{figure}
We perform numerical experiments for $\rho=0$ and $\rho=-0.5$.
For the other parameters, we use again the default parameters from
Table~\ref{defaulttable}. The results are shown in
Figure~\ref{fig:numstability}. For both cases, $\rho=0$ and $\rho=-0.5,$ the
errors show a similar behaviour, being slightly larger for
non-vanishing correlation. There is almost no dependence of the error
on the parabolic mesh ratio $k/h^2,$ which confirms numerically
regular solutions can be obtained without restriction on the time 
step size. For larger values of $h,$ which also result
in a higher cell Reynolds number, the error grows gradually, and no
oscillation in the numerical solutions occurs.
Based on these results and the findings in \cite{DuFo12}, we
conjecture that
the stability condition \eqref{eq:vNstab} also holds for general choice of
parameters.

\section{Conclusion}
\label{concsection}

We have presented a new high-order compact finite difference scheme
for option pricing under stochastic volatility that is fourth order
accurate in space and second order accurate in time. We have conducted a
von Neumann stability analysis (for `frozen coefficients' and periodic
boundary data) and proved
unconditional stability for vanishing correlation. 
In our numerical experiments we
observe a stable behaviour also for a general choice of parameters.
Additional numerical tests presented here and the results of subsequent research
reported in \cite{DuFo12} suggest
that the scheme is also von Neumann stable for non-zero correlation. 
In our numerical convergence study we obtain fourth order numerical
convergence for the non-smooth payoffs which are typical in option pricing.

It would be interesting to consider extensions of this scheme to
non-uniform grids and to the American option pricing problem, where early exercise of the option is
possible. 
An approach to the first would be to introduce a transformation of the
partial differential equation from a non-uniform grid to a uniform
grid \cite{Fournie00}. Then our high order compact methodology can be applied to this
transformed partial differential equation. This is, however, not
straight-forward as the derivatives of the transformation appear in
the truncation error and due to the presence of the cross-derivative
terms. One cannot proceed to cancel terms in the truncation error in a
similar fashion as in the current paper, and the derivation of a
high-order compact scheme becomes much more involved. 
For the second extension, the American option pricing problem, one has to solve a free
boundary problem. It can be
written as a linear complementarity problem which can be
discretised using the scheme \eqref{eq:hocscheme}.
To retain
the high-order convergence one would need to combine the high-order
discretization with a high-order resolution of the free boundary.
Both extensions are beyond the scope of the present paper, and we
leave them for future research.

\appendix

 \bigskip\noindent{{\bf Acknowledgement.}\newline
Bertram D{\"u}ring acknowledges partial support from the
 Austrian Science Fund (FWF), grant P20214, and
from the Austrian-Croatian Project HR 01/2010 
of the Austrian Exchange Service (\"OAD).
The authors are grateful to the anonymous referees for helpful remarks and suggestions.}


\end{document}